\documentclass[conference]{IEEEtran}

\IEEEoverridecommandlockouts                        

\usepackage{enumitem}
\usepackage{cite}
\usepackage{amsmath, amssymb}
\usepackage{xcolor}
\usepackage{algorithm}
\usepackage{algpseudocode}
\usepackage{amsthm}    
\usepackage{amsmath}
\usepackage{amssymb} 
\newtheorem{theorem}{Theorem}

\newtheorem{proposition}{Proposition}
\usepackage{tikz}
\usetikzlibrary{positioning, shapes}
\usepackage{mathtools} 
\usepackage{todonotes}

\newcommand{\R}{\mathbb{R}}

\usepackage{amsthm}  
\usepackage{amsmath}
\usepackage{amssymb}
\usepackage{xcolor}
\usepackage{graphicx} 
\usepackage[framemethod=TikZ]{mdframed}
\usepackage{mdframed}
\usepackage{algorithm}
\theoremstyle{remark}
\newtheorem{remark}{Remark}
\usepackage{algpseudocode}
\usetikzlibrary{arrows.meta,positioning}
\makeatletter
\renewcommand{\ALG@beginalgorithmic}{\setcounter{ALG@line}{0}}
\makeatother
\usepackage{mathtools}

\usepackage{booktabs}
\usepackage{siunitx}
\usepackage{makecell} 
\usepackage{booktabs,siunitx}
\usepackage{graphicx}
\usepackage[caption=false,font=footnotesize]{subfig}

\newcommand{\bN}{\mathbb{N}}
\newcommand{\E}{\mathbb{E}}

\title{\LARGE \bf
Linear Reformulation of Event-Triggered LQG Control under Unreliable Communication
}

\author{Zahra Hashemi and Dipankar Maity
 \thanks{
 This research is supported by the National Science Foundation CAREER Award 2443349.}
\thanks{
The authors are with the Department of Electrical and Computer Engineering at the University of North Carolina at Charlotte, NC, USA, 28223. (e-mails: {\tt \{zahrahashemi1, dmaity\}@charlotte.edu}).}
}

\newcommand{\bl}[1]{{\color{blue}#1}}

\begin{document}

\maketitle
\thispagestyle{empty}
\pagestyle{empty}

\begin{abstract}
We consider event-triggered linear-quadratic Gaussian (LQG) control when sensor updates are transmitted over an i.i.d.\ packet–erasure channel. Although the optimal controller in a standard LQG setup is available in closed form, choosing when to transmit remains computationally and analytically difficult because packetdrops randomize packet delivery and couple scheduling decisions with the estimation-error dynamics, making direct dynamic-programming solutions impractical. By certainty equivalence, the co-design problem becomes choosing a binary send/skip sequence that balances control performance and communication cost. We derive a closed-form expansion of the error covariance as precomputable Gramian terms scaled by a survival factor that depends only on the number of transmission attempts on each interval. This converts the problem into an unconstrained binary program that we linearize exactly via running attempt counters and a one-hot encoding, yielding a compact MILP well-suited to receding-horizon implementation. On the linearized Boeing-747 benchmark, a model predictive control (MPC) scheduler lowers cost while attempting far fewer transmissions than a one-shot baseline across channel success rates.

\end{abstract}

\section{INTRODUCTION}
Networked control systems (NCSs) connect sensors, schedulers, controllers, and actuators over shared communication links. This architecture enables modular design and large-scale deployment, yet it imposes strict limits on bandwidth, latency, and energy. Transmitting measurements at every sampling instant is often infeasible. Event-triggered control (ETC) addresses this by sending information only when it is valuable for closed-loop performance; see, e.g., \cite{heemels2012introduction,maity2019optimal,kostina2019rate,molin2009lqg,suthar2025fly,molin2012optimality,maity2021multiagent,maity2021optimalb,mamduhi2025network}.

Designing optimal ETC policies is challenging because, even without packet loss, the scheduler  and the remote controller  operate with different information sets. 
Random erasures further worsen this mismatch and randomize when estimation-error resets occur. This information asymmetry blocks classical dynamic programming and yields a high-dimensional, nonconvex stochastic problem. Deterministic trigger rules guarantee stability but are generally suboptimal \cite{rabi2012adaptive,lipsa2011remote}; value-of-information based optimization approaches directly model the trade-off between quadratic estimation/control performance and communication effort but suffer from the curse of dimensionality \cite{soleymani2021value}, and continuous-time Hamilton–Jacobi–Bellman (HJB) formulations face similar scalability limits \cite{thelander2020lqg}. In the LQG setting, certainty equivalence reduces co-design to a bilinear scheduling problem with stochastic error dynamics \cite{molin2009lqg,molin2012optimality}, yet the optimization remains difficult for these reasons.

Packet arrivals reset the estimation error to zero while misses let it propagate through the plant and noise, so the covariance recursion couples binary scheduling decisions with error matrices in a bilinear way. For a \emph{lossless} channel, our recent work \cite{hashemi2025linear} shows this problem admits an exact \emph{linear} MILP reformulation via a closed-form covariance expansion and exact linearization of binary products. Introducing packet loss breaks that linear structure via a nonlinear survival factor, motivating the reformulation developed here.

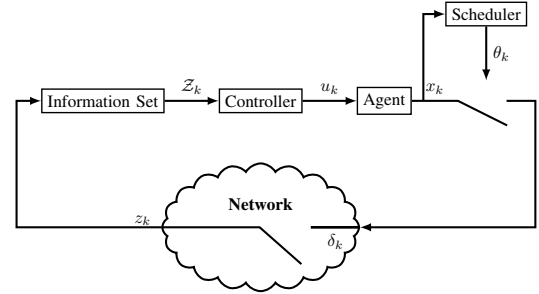
\begin{figure}[t]
\centering
\resizebox{0.8\columnwidth}{!}{
\begin{tikzpicture}[auto, node distance=1.1cm, >=latex, thick]
  \node[draw, rectangle] (Info) {Information Set};
  \node[draw, rectangle, right=of Info] (C) {Controller};
  \node[draw, rectangle, right=of C] (A) {Agent};
  \node[draw, circle, right=of A, minimum size=8mm, inner sep=0pt, color=white] (mult) {};
  \node[draw, rectangle, above=of mult] (Scheduler) {Scheduler};
  \node[draw, very thick,
        cloud, cloud puffs=20, cloud puff arc=120, cloud ignores aspect,
        minimum width=4cm, minimum height=2.6cm,
        below=1cm of C] (Network) {{\raisebox{6ex}{\textbf{Network}}}};
  \draw[->, very thick] (Info) -- node[midway,above] {$\mathcal{Z}_k$} (C);
  \draw[->, very thick] (C) -- node[midway,above] {$u_k$} (A);
  \draw[-, very thick] (A) -- node[midway,above] {$x_k$} +(1.5,0) 
            coordinate[pos=0.25] (branchX) -- +(2.5,-.5);
  \draw[->, very thick] (branchX) |- (Scheduler.west);
  \draw[->, very thick] (Scheduler.south) -- node[right] {$\theta_k$} (mult.north);
  \draw[->, very thick] (mult) -- +(1,0) |- (Network.east);
  \draw[->, very thick] (Network.west) node[left=2pt,yshift=4pt] {$z_k$}
  -| ([yshift=-1cm, xshift=-5mm]Info.south west) |- (Info.west);
  \draw[-, very thick] (Network.west) -- +(2,0) 
            coordinate[pos=0.25] (branchX) -- +(2.85,-0.75);
  \draw[-, very thick] (Network.east) -- node[midway,below] {$\delta_k$}  +(-1,0) |- (Network.east);
\end{tikzpicture}
}

\caption{\footnotesize Architecture with event scheduling and packet erasures. 
The scheduler selects $\theta_k$; the network applies the random success switch $\delta_k$. The controller receives $z_k = x_k$ if $\theta_k \delta_k = 1$, and $z_k = \varnothing$ otherwise. The control input $u_k$ is computed from the information set $\mathcal{Z}_k$. }
\label{fig:system-configuration}
\end{figure}

We study a networked control system in which a sensor measures the process state, a scheduler observes it and issues a binary send/skip decision \(\theta_k\), and a remote controller receives updates over an unreliable channel (see Fig.~\ref{fig:system-configuration}). The channel is modeled as an i.i.d.\ erasure channel with success indicator \(\delta_k\in\{0,1\}\); a fresh measurement is delivered iff an attempt is made and the transmission succeeds, i.e., when \(\theta_k\delta_k=1\). Thus uncertainty stems from both process noise and random packet loss. Our objective is to quantify and optimize the trade-off between control performance and communication effort in this setting. Communication effort is quantified by the expected count of attempts (not necessarily all of them are successful); every attempted transmission incurs the same penalty.

In this work, we develop a scheduler-centric framework for event-triggered LQG over an unreliable i.i.d.\ erasure channel. Using certainty equivalence, we reduce co-design to minimizing a quadratic control cost plus a communication penalty over binary send/skip decisions, where performance depends on whether an attempted packet is actually received. We express the error covariance in closed form as precomputable Gramian terms multiplied by a survival factor determined by the number of transmissions in each interval, yielding an unconstrained binary (nonlinear) optimization; introducing running counters and one-hot encodings converts this into a compact mixed-integer linear program (MILP). We further show that heterogeneous penalties for successful versus failed attempts collapse to a single effective charge without changing the formulation, and we derive one-step send/skip certificates for erasure channels that reduce solving the MILP. 
In addition, we establish schedule-dependent upper and lower ratio bounds of the optimal erasure-channel performance relative to the lossless benchmark, providing explicit, computable guarantees and quantifying the performance impact of packet loss.

 The linear reformulation yields fast computation times and fits with a MPC scheme. On a linearized Boeing–747 benchmark, the resulting MPC policy achieves lower closed-loop cost with significantly less transmission attempts  across a range of success probability.

The remainder of the paper is organized as follows. Section~\ref{sec:problem} introduces the problem setup and modeling assumptions. Section~\ref{sec:control} reviews the certainty-equivalent controller. Section~\ref{sec:communication} develops the event-triggered scheduling policy and its linear-programming reformulation. Section~\ref{sec:case study} reports numerical results, and section~\ref{sec:conclusion} concludes.

\textit{Notation.} $\mathbb{R}$ and $\mathbb{N}_0$ denote the set of reals and the nonnegative integers.
We define \(\|x\|_{Q}^{2}:=x^{\top}Qx\) for $x$ and $Q$ of compatible dimensions.
The trace and transpose of a matrix is denoted by \(\operatorname{tr}(\cdot)\) and ${}^\top$, respectively.
We use the notation $M \succeq 0$ (or $\succ 0$) to denote $M$ to be a positive semi-definite (or, positive definite) matrix.
Unless otherwise stated, expectations \(\mathbb{E}[\cdot]\) are taken with respect to the initial state, process noise and, the channel randomness.

\section{Problem Formulation}\label{sec:problem}
We consider the discrete–time linear system
\begin{equation}
    x_{k+1} = A x_k + B u_k + w_k,\qquad k\in\mathbb{N}_0,
    \label{eq:discrete-system}
\end{equation}
where $x_k\in\mathbb{R}^n$ is the state, $u_k\in\mathbb{R}^m$ is the control input, and
$w_k\in\mathbb{R}^n$ is the process disturbance; $A\in\mathbb{R}^{n\times n}$ and
$B\in\mathbb{R}^{n\times m}$ are known matrices. The disturbance sequence $\{w_k\}_{k\in \bN_0}$ is i.i.d., zero mean, with covariance $\Sigma^w$. The initial state $x_0$ has finite mean and covariance and is independent of the disturbance sequence.

At each time $k$, the sensor observes $x_k$ perfectly and a scheduler chooses
$\theta_k\in\{0,1\}$: $\theta_k=1$ attempts a transmission and $\theta_k=0$ skips it.
The channel is unreliable, so an attempt may fail. Let $\delta_k\in\{0,1\}$ denote
the (random) success indicator. Conditional on an attempt, $\Pr(\delta_k=1\mid\theta_k=1)=p$
and $\Pr(\delta_k=0\mid\theta_k=1)=1-p$, with $\delta_k$ i.i.d.\ across $k$ and
independent of $(x_0,\{w_k\}_{k\in \bN_0})$. When $\theta_k=0$, no transmission is made and $\delta_k$
is irrelevant. See Fig.~\ref{fig:iid-erasure} for a schematic of the i.i.d.\ erasure model.
The controller’s received measurement is
\begin{equation}
z_k =
\begin{cases}
x_k, & \text{if   }~~~\theta_k\,\delta_k=1,\\[3pt]
\varnothing, & \text{otherwise},
\end{cases}
\label{eq:channel}
\end{equation}
so at each step the controller either receives the current state or an erasure.  

We formalize the controller’s knowledge at time $k$ by its information set
\begin{align} \label{eq:infoSet}
\mathcal{Z}_k := \{z_{0:k},\,u_{0:k-1}\},\quad k\in 
\bN,\qquad \mathcal{Z}_0 := \{z_0\},
\end{align}
with the recursive update
\[
\mathcal{Z}_{k+1} = \mathcal{Z}_k \cup \{z_{k+1},\,u_k\}.
\]
The controller does not observe $\theta_k$ or $\delta_k$ directly; their influence is only through the realized $z_k$. If acknowledgments (ACK/NACK) were available, they could be appended to $\mathcal{Z}_k$, but our baseline model uses the definition above.

We denote the entire control input by $U = \{u_k\}_{k=0}^{T-1}$  and the scheduling sequence by $\Theta=\{\theta_k\}_{k=0}^{T-1}$; the stage actions are $u_k$ and $\theta_k$. The goal is to jointly determine $U$ and $\Theta$ that minimize the expected finite-horizon cost
\begin{equation}
J(U,\Theta)\!=\! \mathbb{E}\!\left[\!\sum_{k=0}^{T-1} 
\!\!\big(\|x_k\|_Q^2 + \|u_k\|_R^2 + \lambda \theta_k\big)\! + \|x_T\|_{Q_T}^2\right]\!\!.
\label{eq:cost}
\end{equation}
The matrices $Q \succeq 0$ and $Q_T \succeq 0$ penalize state deviations, $R \succ 0$ penalizes control effort, and $\lambda > 0$ represents the communication cost.

In this formulation, the communication penalty is associated with the scheduling decision $\theta_k$ 
rather than whether the communication was successful (i.e.,  $\delta_k \theta_k$) 
since network resources are allocated whenever a transmission is attempted, regardless of its success. 
If one wishes to distinguish between the costs of successful and failed transmissions, 
the model can be extended accordingly, as detailed in Remark~\ref{rem:heterogeneous}.
\begin{figure}[t]
\centering
\begin{tikzpicture}[
    scale=0.9, every node/.style={transform shape},
    >=Stealth,
    state/.style={circle,draw,semithick,minimum size=10mm,inner sep=0pt,font=\bfseries\small},
    every loop/.style={looseness=8}
]
\node[state,fill=red!40!white] (z) {0};
\node[state,fill=green!60!white,right=2.2cm of z] (o) {1};
\path[->] (z) edge[loop above] node {$1-p$} (z);
\path[->] (o) edge[loop above] node {$p$} (o);
\path[->] (z) edge[bend left=18] node[above] {$p$} (o);
\path[->] (o) edge[bend left=18] node[below] {$1-p$} (z);
\end{tikzpicture}
\caption{\footnotesize Packet–erasure channel: state 1 = success, 0 = erasure. For each time-step $k$, packer-drop is modeled by an i.i.d. Bernoulli random variable $\delta_k$ with success probability $p$.}
\label{fig:iid-erasure} \vspace{-2mm}
\end{figure}
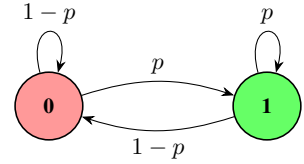

\section{Optimal Controller} \label{sec:control}
The optimal control law turns out to be a linear controller of the type of certainty equivalence (independent of the scheduler), while the scheduler remains to be designed; thus, the joint problem reduces to selecting a binary send/skip policy that shapes the stochastic estimation error \cite{molin2009lqg,molin2012optimality}. For a finite horizon \(T\), the optimal controller is given by
\begin{equation}
u_k = -L_k\,\mathbb{E}[x_k \mid \mathcal{Z}_k], \qquad k=0,\ldots,T-1,
\label{eq:molin-control}
\end{equation}
where $\mathcal{Z}_k$ is the information available to the controller (see \eqref{eq:infoSet}). The gains $\{L_k\}_{k=0}^{T-1}$ are computed by a Riccati recursion:
\begin{subequations}
\begin{align}
L_k &= S_k^{-1} B^\top P_{k+1} A, \\
S_k &= R + B^\top P_{k+1} B, \\
P_k &= A^\top P_{k+1} A + Q - A^\top P_{k+1} B S_k^{-1} B^\top P_{k+1} A,
\end{align}
\end{subequations}
with terminal condition $P_T = Q_T$.

For convenience, define the effective reception indicator $\tilde{\theta}_k := \theta_k \delta_k$, which equals~$1$ when an attempted transmission succeeds. The conditional state estimate evolves as follows~\cite{molin2012optimality}:
\[
\hat{x}_k := \mathbb{E}[x_k \mid \mathcal{Z}_k]
=
\begin{cases}
x_k, & \tilde{\theta}_k = 1,\\[4pt]
A\hat{x}_{k-1} + B u_{k-1}, & \tilde{\theta}_k = 0.
\end{cases}
\]

Substituting~\eqref{eq:molin-control} into the cost functional reduces the joint problem to an equivalent optimization over the transmission sequence $\Theta$:
\begin{equation}
J(U^*,\Theta)
= J_{\text{const}}
+ \mathbb{E}\!\Bigg[\sum_{k=0}^{T-1} \|e_k\|_{L_k^\top S_k L_k}^2\Bigg]
+ \lambda\,\mathbb{E}\!\Bigg[\sum_{k=0}^{T-1}\theta_k\Bigg],
\label{eq:cost-function2}
\end{equation}
where $e_k := x_k - \hat{x}_k$ is the estimation error, and
\[
J_{\text{const}}
:= \mathbb{E}[x_0^\top P_0 x_0]
+ \sum_{k=0}^{T-1} \mathbb{E}[w_k^\top P_{k+1} w_k],
\]
which depends only on the initial state and the noise sequence, and is therefore independent of~$\Theta$. Finally, the estimation-error sequence obeys the recursion
\begin{equation}
e_{k+1} = \bigl(1-\tilde{\theta}_{k+1}\bigr)\bigl(A e_k + w_k\bigr), 
\qquad e_{-1} := 0.
\label{eq:error-dynamics}
\end{equation}
Thus, whenever a new packet is successfully delivered ($\tilde{\theta}_{k+1}=1$), the error resets to zero; otherwise, it propagates under the open-loop system dynamics with additive process~noise.

\section{Communication Protocol}\label{sec:communication}
We next formulate the scheduling task in receding–horizon form. 
Define the scheduler–side error
\begin{align*}
    e^s_k &:= x_k - A \hat{x}_{k-1} - B u_{k-1} = A e_{k-1} + w_{k-1},\\
    e^s_0 &:= x_0 - \mathbb{E}[x_0].
\end{align*}
The controller’s estimation error satisfies
\[
e_k = \bigl(1-\tilde{\theta}_{k}\bigr)\,e_k^s, \qquad \,k\in \bN_0.
\]

With weights $\Gamma_t := L_t^\top S_t L_t$, the scheduler at time $k$ plans the
send/skip decisions over the remaining horizon by solving:
\begin{align}
\label{eq:MPC_optt}
\min_{\{\theta_t\}_{t=k}^{T-1}}\quad 
& \mathbb{E}\!\left[\sum_{t=k}^{T-1}\|e_t\|_{\Gamma_t}^2 + \lambda\,\theta_t\right] \nonumber\\
\text{s.t.}\quad 
& e_{t+1} = \bigl(1-\tilde{\theta}_{t+1}\bigr)\bigl(A e_t + w_t\bigr), \\
& e_k = \bigl(1-\tilde{\theta}_{k}\bigr)\,e_k^s. \nonumber
\end{align}
Let
\[
\Theta_k^*(e_k^s)
:= \{\theta_{k|k}^*,\,\theta_{k+1|k}^*,\,\ldots,\,\theta_{T-1|k}^*\}
\]
denote an optimal schedule computed at time $k$ for the current error $e_k^s$.
A receding–horizon implementation applies only the first decision, $\theta_{k|k}^*$,
computes the scheduler-side error $e_{k+1}^s$, and re-solves \eqref{eq:MPC_optt} at the next step.

To design a computationally tractable procedure for solving the stochastic optimization problem~\eqref{eq:MPC_optt}, 
we proceed by evaluating its objective under a given sequence of scheduling decisions. 
Let $\Theta_k = \{\theta_{k|k}, \ldots, \theta_{T-1|k}\}$ denote a fixed decision rollout. 
We define $\bar{J}(\Theta_k, k \mid e^s_k)$ as the corresponding cost associated with~\eqref{eq:MPC_optt}, namely,
\begin{align}
    \bar{J}(\Theta_k, k \mid e^s_k)
    := &\E\left[\sum_{t=k}^{T-1} \|e_t\|^2_{\Gamma_t} + \lambda \theta_t~\bigg|~\Theta, e^s_k\right] \nonumber \\
    =&\sum_{t=k}^{T-1} \big(\operatorname{tr}(\Gamma_t \Sigma_{t|k}) + \lambda \theta_{t|k}\big),
    \label{eq:MPC_cost1}
\end{align}
where $\Sigma_{t|k} = \mathbb{E}[e_t e_t^\top \mid \Theta_k, e^s_k]$ for $t = k, \ldots, T-1$. 
\begin{align*}
&\Sigma_{t+1|k}
= \mathbb{E}[e_{t+1} e_{t+1}^\top~|~\Theta_k, e^s_k] \\[4pt]
&= \mathbb{E}\!\left[(1 - \tilde{\theta}_{t+1|k})^2 (A e_t + w_t)(A e_t + w_t)^\top~|~\Theta_k, e^s_k\right] \\[4pt]
&\overset{(\dagger)}{=}(1 - {\theta}_{t+1|k}\mathbb{E}[\delta_{t+1}])
   \big(A\,\mathbb{E}[e_t e_t^\top~|~\Theta_k, e^s_k]A^\top + \mathbb{E}[w_t w_t^\top]\big) \\[4pt]
&=(1 - p\,\theta_{t+1|k})\!\left(A \Sigma_{t|k} A^\top + \Sigma^w\right) %\\[4pt]
%&\overset{(\S)}{=}(1 - p)^{\theta_{t+1|k}}\!\left(A \Sigma_t A^\top + \Sigma^w\right).
\end{align*}
Here, $(\dagger)$ follows from the facts that (i) $(1 - \tilde{\theta}_{t+1|k})^2 = 1 - \tilde{\theta}_{t+1|k}$ and ${\delta}_{t+1}$ is independent of the randomness of $e_t, w_t$, and (ii) $w_t$ has zero mean and is independent of $e_t$.

For $t=k$, $e_k = \bigl(1-\tilde{\theta}_{k}\bigr)\,e_k^s$, and therefore
\[
\Sigma_{k|k} = \mathbb{E}[e_k e_k^\top \mid \Theta_k, e^s_k] = (1 - p{\theta_{k|k}})\, e_k^s (e_k^s)^\top.
\]
Hence, the original stochastic optimization problem can be equivalently expressed as a deterministic \emph{mixed-integer nonlinear program (MINLP)} characterized by bilinear matrix equality constraints.

\begin{proposition}[Matrix–recursion formulation of optimal scheduling]
\label{prop:matrix-form}
For any initial innovation $e_k^s$, the optimal scheduling problem over the horizon $[k,T)$ is equivalent to the deterministic program
\begin{align}
\label{eq:MPC_2}
\min_{\Theta,\,\Sigma} \quad & \bar{J}(\Theta, k \mid e_k^s) \nonumber \\
\text{\normalfont s.t.} \quad
& \Sigma_{t+1\mid k} = \bigl(1 - p\,\theta_{t+1\mid k}\bigr)\!\left(A \Sigma_{t\mid k} A^\top + \Sigma^w\right), \nonumber \\[2pt]
& \hspace{5.2em} \forall\, t = k,\ldots,T-2, \nonumber \\
& \Sigma_{k\mid k} = \bigl(1 - p\,\theta_{k\mid k}\bigr)\, e_k^{s}(e_k^{s})^\top .
\end{align}
\end{proposition}

Despite the linearity of the objective with respect to $\{\Sigma_{t\mid k}\}_{t=0}^{T-1}$ and $\Theta_k$, the bilinear coupling in the matrix equalities considerably complicates the optimization. This coupling reveals that packet losses directly increase estimation uncertainty, which in turn affects \emph{when} and \emph{how often} transmissions are triggered. The number of decision variables grows with the horizon length and state dimension, posing scalability challenges; nevertheless, the problem structure is well suited to the mixed-integer optimization approach presented next.

To address these challenges, we derive a closed-form expression for the covariance, eliminating the recursive dependence and reformulating the problem as an unconstrained binary optimization program.

\begin{proposition}[Closed-Form Covariance and Unconstrained Scheduling Program]
\label{thm:closed-form}
Let $p\in (0,1)$. For a horizon starting at time~$k$, the estimation-error covariance admits the closed-form decomposition 
\begin{equation}
\Sigma_{t|k}
= \sum_{\tau=k}^t 
   \big( (1 - p)^{\sum_{s=\tau}^t \theta_{s|k}}\big)
   G_{t,\tau},
\label{eq:Sigma-sum}
\end{equation}
where 
\begin{equation}
G_{t,\tau} =
\begin{cases}
A^{t-k}\,e_k^s(e_k^s)^\top (A^{t-k})^\top, & \tau = k, \\[4pt]
A^{t-\tau}\Sigma^w(A^{t-\tau})^\top, & \tau \ge k+1.
\end{cases}
\end{equation}
Define
\begin{equation}
g_{t,\tau} := \operatorname{tr}(\Gamma_t G_{t,\tau}).
\end{equation}
Consequently, the finite-horizon scheduling problem reduces to the unconstrained binary program
\begin{equation}
\min_{\Theta_k}
\;\sum_{t=k}^{T-1} \sum_{\tau=k}^t 
   (1 - p)^{\sum_{s=\tau}^t \theta_{s|k}} g_{t,\tau}
   + \lambda \sum_{t=k}^{T-1} \theta_{t|k}.
\label{eq:unconstrained}
\end{equation}
\end{proposition}
\begin{proof}
We proceed by induction on $t$.  
For $t = k$, \eqref{eq:MPC_2} can be written as 
\[
\Sigma_{k|k} = (1 - p\,\theta_k)\, e_k^s (e_k^s)^\top = (1 - p)^{\theta_k}\, e_k^s (e_k^s)^\top,
\]
Assume the representation holds at time $t$, that is,
\[
\Sigma_{t|k} = \sum_{\tau = k}^t 
  \big( (1 - p)^{\sum_{s=\tau}^t \theta_{s|k}}\big) G_{t,\tau}.
\]
Substituting this expression into~\eqref{eq:MPC_2} yields 
{\small
\begin{align*}
\Sigma_{t+1|k}
&= (1 - p {\theta_{t+1|k}}\,)\!\left(A \Sigma_{t|k} A^\top + \Sigma^w\right) \\[3pt]
&\overset{(*)}{=}(1 - p \,)^{{\theta_{t+1|k}}}\!\left(A \Sigma_{t|k} A^\top + \Sigma^w\right)  \\[3pt]
&= \Bigg( A \Big(\sum_{\tau = k}^t 
      \big( (1 - p \,)^{\theta_{t+1|k}}(1 - p)^{\sum_{s=\tau}^t \theta_{s|k}}\big) G_{t,\tau}\Big) A^\top  \\[3pt]
  & \qquad + (1 - p \,)^{\theta_{t+1|k}}\Sigma^w \Bigg) =\sum_{\tau = k}^{t+1} 
     \big( (1 - p)^{\sum_{s=\tau}^{t+1} \theta_{s|k}}\big)
    G_{t+1,\tau}.
\end{align*}
}Step $(*)$ holds because $(1-p\,\theta_{t+1|k})$ equals $1$ when $\theta_{t+1|k}=0$ and equals $(1-p)$ when $\theta_{t+1|k}=1$. This completes the induction. The instantaneous stage cost (see \eqref{eq:MPC_cost1}) therefore simplifies to
\[
\operatorname{tr}(\Gamma_t \Sigma_{t|k})
  = \sum_{\tau = k}^t
  (1 - p)^{\sum_{s=\tau}^t \theta_{s|k}} 
      g_{t,\tau},
\]
where $g_{t,\tau}=\operatorname{tr}(\Gamma_t G_{t,\tau})$ can be precomputed offline.  
Substituting this expression into \eqref{eq:MPC_cost1} completes the proof.
\end{proof}

The only nonlinearity in \eqref{eq:unconstrained} is the term $ (1 - p)^{\sum_{s=\tau}^t \theta_{s|k}}$. Let \(c_{t,\tau} := \sum_{s=\tau}^t \theta_{s|k}\) denote the number of transmission attempts on \([\tau,t]\). 
To obtain a mixed–integer \textit{linear formulation}, we introduce running counters \(c_{t,k}\in\mathbb{Z}_{\ge0}\) satisfying
\[
c_{t,k} = c_{t-1,k} + \theta_{t|k},\qquad c_{k-1,k}=0,
\]
so that \(c_{t,k}=\sum_{s=k}^t \theta_{s|k}\) and \(0\le c_{t,k}\le t-k+1\).
Interval counts then follow as
\[
c_{t,\tau} = c_{t,k} - c_{\tau-1,k},\qquad k\le \tau \le t \le T-1,
\]
with the bounds \(0\le c_{t,\tau}\le t-\tau+1\).
Each \(c_{t,\tau}\) is encoded by one-hot binary selectors \(s_{t,\tau,i}\in\{0,1\}\) for \(i=0,\dots,t-\tau\), enforced by
\[
\sum_{i=0}^{t-\tau} s_{t,\tau,i}=1,\qquad \sum_{i=0}^{t-\tau} i\,s_{t,\tau,i}=c_{t,\tau}.
\]
Let \(\beta_i:=(1-p)^i\) (precomputed constants depending only on \(p\)).
Then
\[
(1-p)^{c_{t,\tau}}=\sum_{i=0}^{t-\tau}\beta_i\,s_{t,\tau,i}.
\]
Substituting this identity into the closed-form cost yields the mixed–integer linear program stated next.
\begin{mdframed}[backgroundcolor=gray!10, roundcorner=10pt,
  innerleftmargin=4pt, innerrightmargin=4pt,
  innertopmargin=4pt, innerbottommargin=4pt]

\noindent\textbf{Final MILP Reformulation:}
\begin{align}
\min_{\theta,\,c,\,s} \quad &
\sum_{t=k}^{T-1} \sum_{\tau=k}^t \sum_{i=0}^{t-\tau} 
  \beta_i\, g_{t,\tau}\, s_{t,\tau,i}
  + \lambda \sum_{t=k}^{T-1} \theta_{t|k}, 
  \label{eq:milp-obj} \\[4pt]
\text{s.t.}\quad &
c_{t,k} = c_{t-1,k} + \theta_{t|k}, 
\qquad c_{k-1,k} = 0, \nonumber\\
& c_{t,\tau} = c_{t,k} - c_{\tau-1,k}, 
\qquad k \le \tau \le t \le T-1, \nonumber\\
& 0 \le c_{t,\tau} \le t - \tau + 1, \nonumber\\
& \sum_{i=0}^{t-\tau} s_{t,\tau,i} = 1, 
\qquad 
  \sum_{i=0}^{t-\tau} i\, s_{t,\tau,i} = c_{t,\tau}, \nonumber\\
& \theta_{t|k} \in \{0,1\}, 
\quad s_{t,\tau,i} \in \{0,1\}, 
\quad c_{t,\tau} \in \mathbb{Z}_{\ge 0}. \nonumber
\end{align}
\end{mdframed}

The MILP \eqref{eq:milp-obj} is to be solved at every time-step $k$ based on the realized error $e^s_k = x_k - (Ax_{k-1}+Bu_{k-1})$. 
However, we will soon show (Theorem~1) that there exists an ellipsoid $\mathcal{E}_k^{\rm skip}$ such that for all $e^s_k \in \mathcal{E}_k^{\rm skip}$ we can certify without solving the MILP that the optimal decision at time $k$ is to skip (i.e., $\theta_{k|k}^* =0$), thus saving some computational burden. 
Similarly, there exists another ellipsoid $\mathcal{E}_k^{\rm attempt} \supset \mathcal{E}_k^{\rm skip}$ such that for all $e^s_k \notin \mathcal{E}_k^{\rm attempt}$, the optimal strategy at that time is to attempt (i.e., $\theta_{k|k}^* =1$). 
We provide analytical characterization of these ellipsoids and use them as a one-step optimality certificate and solve the MILP only when $e^s_k \in \mathcal{E}_k^{\rm attempt} \setminus \mathcal{E}_k^{\rm skip}$. 

\begin{theorem}[One-step optimality certificates]
\label{thm:one-step-certificates-erasure}
Define
\[
W_k := \sum_{j=0}^{T-1-k} (A^j)^\top \Gamma_{k+j} A^j .
\]
Let $\mathrm{Ben}_k(e_k^s)$ denote the expected gain/loss from not attempting to communicate over making an attempt at time $k$. Then,
\begin{equation}
\label{eq:sandwich-erasure}
p\,e_k^{s\top}\Gamma_k e_k^s - \lambda
\;\le\;
\mathrm{Ben}_k(e_k^s)
\;\le\;
p\,e_k^{s\top} W_k e_k^s - \lambda .
\end{equation}
Consequently, if $p\,e_k^{s\top}\Gamma_k e_k^s \ge \lambda$ then \emph{attempt} is (weakly) optimal, and
if $p\,e_k^{s\top} W_k e_k^s \le \lambda$ then \emph{skip} is (weakly) optimal.
\end{theorem}

\noindent\emph{Proof.} See Appendix~\ref{app:proof-one-step}. 

The `weakly' part in the theorem statement implies that when the inequalities are held with equalities (e.g., $p\,e_k^{s\top}\Gamma_k e_k^s = \lambda$) both $\theta_{k|k}^* = 0$ and  $\theta_{k|k}^* = 1$ result in the same expected performance, therefore skip (or equivalently, attempt) is weakly optimal over attempt (or equivalently, skip). 
Otherwise, when strict inequalities hold, one action (skip/attempt) is strictly optimal over another. 
From Theorem~\ref{thm:one-step-certificates-erasure} we obtain 
\begin{align*}
    &\mathcal{E}_k^{\rm skip} = \{e \in \R^n~|~ \|v\|^2_{\frac{p}{\lambda}W_k} \le 1\},\\
    &\mathcal{E}_k^{\rm attempt} = \{e \in \R^n~|~ \|v\|^2_{\frac{p}{\lambda}\Gamma_k} \le 1\}.
\end{align*}

As $p$ decreases or $\lambda$ increases, the skip set $\mathcal{E}_k^{\rm skip}$ increases, illustrating the fact that skip is optimal in a larger region. 
In the limit as $\tfrac{p}{\lambda} \to 0$, we have $\mathcal{E}_k^{\rm skip} \to \R^n$, indicating that transmission should never be attempted, as one would expect for the case of a highly unreliable channel ($p\to 0$) and highly costly communication ($\lambda \to \infty$). 
Analogous conclusions can be drawn for the case where $\tfrac{p}{\lambda} \to \infty$ (or, $\lambda \to 0$) where ``attempt all the time" becomes optimal, as one would expect for the case with no communication cost (i.e., $\lambda =0$).

Next, we establish schedule–dependent upper and lower ratio bounds for the optimal performance  relative to the ideal-channel benchmark (i.e., $p=1$), thereby providing explicit, computable guarantees on performance degradation as a result of channel quality $p$.

\begin{theorem}[Ratio bounds relative to the lossless optimum]
\label{thm:ratio-lossless}
For any channel probability $p\in (0,1)$, let $J_p(\Theta_k,k)$ be as in \eqref{eq:unconstrained}, and define
$J_p^{\star}(k):=\min_{\Theta_k} J_p(\Theta_k,k)$ and $J_1^{\star}(k):=\min_{\Theta_k} J_1(\Theta_k,k)$ as optimal performance in lossy and ideal channels, respectively, for the horizon $[k,T]$.
Fix any ideal channel optimizer optimizer $\Theta_1^{\star}\in\arg\min_{\Theta} J_1(\Theta_k,k)$ and any
$p$-optimizer $\Theta_p^{\star}\in\arg\min_{\Theta} J_p(\Theta_k,k)$.
Then,
\begin{equation}
\label{eq:ratio-tight}
\frac{J_p^{\star}(k)}{J_1^{\star}(k)}
\;\le\;
1 + (1-p)\,
\frac{\displaystyle
\sum_{t=k}^{T-1}\sum_{\tau=k}^{t}
\mathbf{1}\!\big\{c_{t,\tau}(\Theta^{\star}_{1})\ge 1\big\}\, g_{t,\tau}}
{\displaystyle J_1^{\star}(k)} .
\end{equation}
Moreover, letting $H:=T-k$ and
$C_{\max}:=\max_{k\le \tau\le t\le T-1} c_{t,\tau}(\Theta_p^{\star})\le H$, the following
\emph{lower} bound also holds:
\begin{equation}
\label{eq:ratio-lower-dependent}
1 +
(1-p)^{C_{\max}}
\tfrac{\displaystyle
\sum_{t=k}^{T-1}\sum_{\tau=k}^{t}
\mathbf{1}\!\big\{c_{t,\tau}(\Theta^{\star}_{p})\ge 1\big\}\, g_{t,\tau}}
{\displaystyle J_1(\Theta_p^{\star},k)}
\le
\frac{J_p^{\star}(k)}{J_1^{\star}(k)}  .
\end{equation}
\end{theorem}

\noindent\emph{Proof.} See Appendix~\ref{app:proof-ratio-lossless}.

As $p\to 1$, $\tfrac{J_p^{\star}}{J_1^{\star}} \to 1$, as expected. 
The upper bound \eqref{eq:ratio-tight}, which is more operationally significant, shows how performance degrades as $p \to 0$. 

The complete workflow, offline precomputations, interval–count encoding, MILP construction, and the receding–horizon loop, is summarized in Algorithm~\ref{alg:et-lqg}.
We conclude our analysis with two remarks highlighting the special cases: (i) when $p=1$ the MILP \eqref{eq:milp-obj} does not work since the relaxation $(1-p\theta) = (1-p)^\theta$ fails, and (ii) when successful and failed attempts have different communication costs.  

 \begin{remark}
When $0<p<1$, the survival factor $(1-p)^{c_{t,\tau}}$ depends on the attempt count
$c_{t,\tau}$, so the one–hot linearization in Section~\ref{sec:communication} is appropriate.
In contrast, when the channel is perfect ($p=1$), the survival factor becomes a Boolean
product of binary complements,
\[
\mu_{t,\tau}(\theta)
= \prod_{s=\tau}^{t}\bigl(1-\theta_{s\mid k}\bigr) \in \{0,1\},
\]
which equals $1$ if and only if no attempts are issued on $[\tau,t]$.
In this case, the nonlinearity collapses to products of binary variables and can be linearized
\emph{exactly} with fewer variables by introducing auxiliary binaries to represent each product.
Because Fortet–McCormick envelopes are exact for binary variables, each weighted term becomes linear,
and the objective reduces to a binary linear program with only $H=T-k$ scheduling binaries
$\{\theta_{t\mid k}\}_{t=k}^{T-1}$ and $\tfrac{H(H+1)}{2}$ auxiliary binaries (one per interval $[\tau,t]$).
A compact formulation for the case $p=1$, including equivalent cascaded McCormick constructions,
is given in~\cite{hashemi2025linear}. \hfill $\blacktriangle$
\end{remark}

\begin{remark}[Heterogeneous communication penalties]
\label{rem:heterogeneous}
In some applications, unsuccessful and successful transmission attempts may incur different costs.  
Let $\lambda_1>0$ denote the penalty for an attempted but \emph{unsuccessful} transmission, 
and $\lambda_2>0$ the penalty for an attempted and \emph{successful} one.  
The expected communication cost over the horizon $[k,T)$ is then
\[
C_{\mathrm{comm}}
= \mathbb{E}\!\left[\sum_{t=k}^{T-1}
\lambda_1\,\theta_{t|k}\,(1-\delta_t) \;+\; \lambda_2\,\theta_{t|k}\,\delta_t \right].
\label{eq:comm-split}
\]
Since $\Pr(\delta_t=1 \mid \theta_{t|k}=1)=p$, it follows that
\[
C_{\mathrm{comm}}
= \mathbb{E}\!\left[\sum_{t=k}^{T-1} \big(\lambda_1(1-p)+\lambda_2 p\big)\,\theta_{t|k}\right]
= \lambda_{\mathrm{eff}}\,\mathbb{E}\!\left[\sum_{t=k}^{T-1}\theta_{t|k}\right],
\label{eq:lambda-eff}
\]
where
\[
\lambda_{\mathrm{eff}} := \lambda_1(1-p)+\lambda_2 p.
\]
Thus, the heterogeneous-penalty model is \emph{equivalent} to the baseline formulation after replacing $\lambda$ with $\lambda_{\mathrm{eff}}$.  
All subsequent derivations—including the covariance recursion, the closed-form expression, and the MILP reformulation—remain valid under this generalization. \hfill $\blacktriangle$
\end{remark}

\begin{algorithm}
\caption{MPC Event–Triggered LQG via MILP}
\label{alg:et-lqg}
\begin{algorithmic}[1]
\small
\Require System $(A,B)$, weights $(Q,Q_T,R)$, noise covariance $\Sigma^w$, success prob.\ $p$, penalty $\lambda$, horizon $T$
\Ensure Schedule $\{\theta_k\}_{k=0}^{T-1}$ and controls $\{u_k\}_{k=0}^{T-1}$

\For{$k=T-1$ \textbf{down to} $0$}
  \State $S_k \gets R + B^\top P_{k+1} B$
  \State $L_k \gets S_k^{-1} B^\top P_{k+1} A$
  \State $P_k \gets A^\top P_{k+1} A + Q - A^\top P_{k+1} B S_k^{-1} B^\top P_{k+1} A$
  \State $\Gamma_k \gets L_k^\top S_k L_k$
\EndFor

\State Precompute $\beta_i:=(1-p)^i$,
       $G_{t,\tau}:=A^{t-\tau}\Sigma^w(A^{t-\tau})^\top$ and
       $g_{t,\tau}:=\operatorname{tr}(\Gamma_t G_{t,\tau})$
       for $k\le\tau\le t\le T-1$, $\tau\ge k{+}1$
\State Precompute tail Gramians
       \[
         W_t \gets \sum_{j=0}^{T-1-t} (A^j)^\top \Gamma_{t+j} A^j,\qquad t=0,\dots,T-1 .
       \]

\For{$k=0$ \textbf{to} $T-1$}
  \State $e_k^s \gets x_k - A\hat{x}_{k-1} - B u_{k-1}$ \Comment innovation
  \State decided $\gets$ \textbf{false}
  \If{$p\, e_k^{s\top}\Gamma_k e_k^s \ge \lambda$} \Comment attempt is optimal
     \State $\theta_k \gets 1$; \quad decided $\gets$ \textbf{true}
  \ElsIf{$p\, e_k^{s\top} W_k e_k^s \le \lambda$} \Comment skip is optimal
     \State $\theta_k \gets 0$; \quad decided $\gets$ \textbf{true}
  \EndIf
  \If{\textbf{not} decided}
     \State Solve MILP~\eqref{eq:milp-obj} on $t=k{:}T-1$ to obtain $\{\theta_t^\star\}$
     \State $\theta_k \gets \theta_k^\star$ \Comment implement first action
  \EndIf

  \If{$\theta_k = 1$}
    \State Transmit; receive ACK/NACK and observe $\delta_k$
  \Else
    \State $\delta_k \gets 0$
  \EndIf

  \If{$\delta_k = 1$}
    \State $\hat{x}_k \gets x_k$ \Comment fresh packet received
  \Else
    \State $\hat{x}_k \gets A\hat{x}_{k-1} + B u_{k-1}$ \Comment prediction step
  \EndIf
  \State $u_k \gets -L_k\,\hat{x}_k$
\EndFor 
\end{algorithmic} 
\end{algorithm}  

\section{Numerical Example}\label{sec:case study}
We evaluate the proposed approach on the linearized longitudinal dynamics of a Boeing~747 in steady, level flight 
(altitude \(40{,}000\,\mathrm{ft}\), speed \(774\,\mathrm{ft/s}\)) with a sampling period of \(1\,\mathrm{s}\); 
see~\cite{boyd2004convex} for modeling details. 
The resulting dynamics are:
\[
x_{k+1} = A x_k + B u_k + w_k,
\]
with
{\small 
\[
A =
\begin{bmatrix}
  0.99 &  ~~\,0.03 & -0.02 & -0.32\\
  0.01 &  ~~\,0.47 &  ~~\,4.70 &  ~~\,0.00\\
  0.02 & -0.06 &  ~~\,0.40 &  ~~\,0.00\\
  0.01 & -0.04 &  ~~\,0.72 &  ~~\,0.99
\end{bmatrix}\!\!,\quad B =
\begin{bmatrix}
  ~~\,0.01 & 0.99\\
 -3.44 & 1.66\\
 -0.83 & 0.44\\
 -0.47 & 0.25
\end{bmatrix}\!\!.
\]
}

The process noise is zero-mean Gaussian with covariance \(\Sigma^w = \sigma^2 I\) (\(\sigma > 0\)). 
The packet is successfully delivered with probability \(p = 0.7\). 
The communication penalty is set to \(\lambda = 100\). 
The LQG weights are chosen as \(Q = 5I\), \(R = I\), and \(Q_T = 5I\). 
Furthermore, \( \E[x_0] =  [0.5,\,0.5,\,0.5,\,0.5]^\top\) and the covariance of the initial state is \(\Sigma_0 =0.4I\). 
Simulations are performed over a horizon of \(T = 50\).

Figs.~\ref{fig:err-oneshot}–\ref{fig:err-mpc} show the state-estimation errors under the \textit{one-shot}, where \eqref{eq:milp-obj} is solved only once at time $0$, and the \textit{MPC scheduling} modes. 
Both runs use the same realization of the i.i.d.\ channel process and the same noise sample paths. 
The MPC scheduler maintains the estimation error closer to zero and reduces the large deviations observed under the one-shot schedule 
by triggering transmissions just before the predicted error growth.

The transmission patterns in Figs.~\ref{fig:sch-oneshot1}–\ref{fig:sch-mpc1} show that the MPC-based scheduler is both sparser and more selective. 
It concentrates transmission attempts after periods of error buildup and skips updates when the predicted error remains small. 
The overall transmission success ratios are comparable (one–shot: \(12/23\); MPC: \(8/15\)), indicating that the performance improvement stems from better timing of transmissions rather than from random channel outcomes.
% --- figure (2x2) ---
\begin{figure}[t]
\centering

\subfloat[One-shot: state error\label{fig:err-oneshot}]{
  \includegraphics[width=0.48\columnwidth,trim=200 325 190 350]{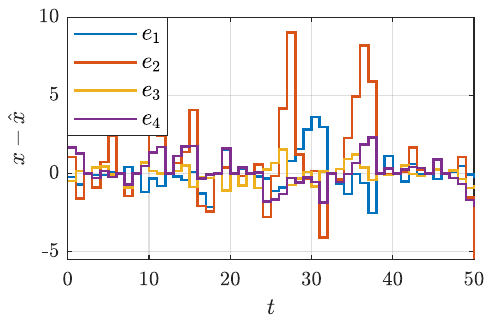}}
\hfill
\subfloat[MPC: state error\label{fig:err-mpc}]{
  \includegraphics[width=0.48\columnwidth,trim=200 325 190 350]{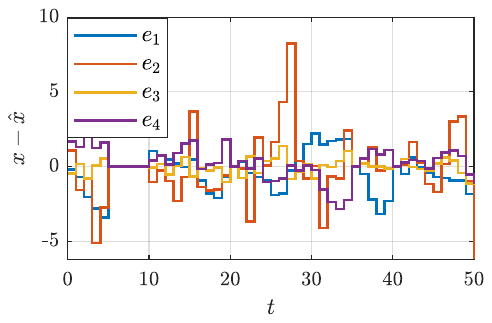}}
\par\medskip

\subfloat[One-shot: scheduling\label{fig:sch-oneshot1}]{
  \includegraphics[width=0.48\columnwidth,trim=200 325 190 350]{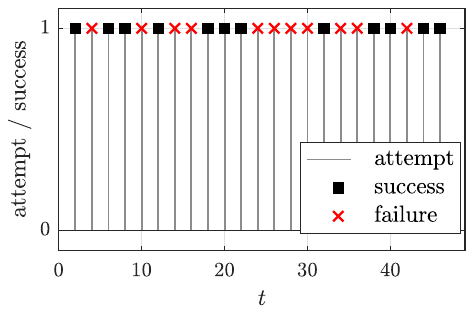}}
\hfill
\subfloat[MPC: scheduling\label{fig:sch-mpc1}]{
\includegraphics[width=0.48\columnwidth,trim=200 325 190 350]{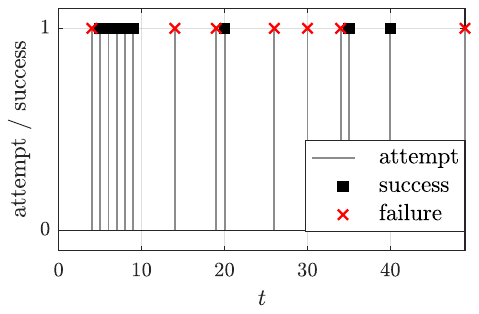}} % (filename as provided)
\caption{State-error predictions and scheduling decisions under one-shot and MPC strategies. }
\label{fig:four-up} \vspace{-.5 cm}
\end{figure}

Table~\ref{tab:comm-costs1} quantifies these effects. 
The MPC scheduler reduces the number of transmission attempts from \(23\) to \(15\) (\(-34.8\%\)), 
lowering the corresponding communication cost from \(2300\) to \(1500\). 
The total realized cost, which combines control performance and communication expenditure, 
decreases from \(7192.96\) to \(5635.19\); a \(21.7\%\) reduction. 
These results are consistent with the averaged Monte Carlo analysis over 100 random seeds, 
confirming the robustness of the observed performance trends (see Table~\ref{tab:comm-costs}).

Fig.~\ref{fig:attempts_over_time} shows the normalized fraction of communication attempts over time, averaged across 100 Monte Carlo trials. 
The blue and red bars correspond to the one-shot and MPC-based schedulers, respectively, and each bar height indicates the proportion of runs in which a transmission was attempted~at that time-step. 
The MPC scheduler triggers communication~less frequently and with greater temporal variation, reflecting its predictive and selective behavior. 
In contrast, the one-shot strategy attempts transmissions almost every step with an almost periodic structure, leading to higher overall communication activity. 
These temporal patterns confirm that the MPC approach achieves comparable performance with substantially fewer attempts, consistent with the aggregate statistics~reported in Table~\ref{tab:comm-costs}.

\begin{table}[h]
\centering
\setlength{\tabcolsep}{4pt}          
\renewcommand{\arraystretch}{1.05}
\footnotesize
\caption{Communication and realized costs.}
\label{tab:comm-costs1}
\begin{tabular}{
  l
  S[table-format=2.0,round-mode=places,round-precision=0] % integers
  S[table-format=4.0,round-mode=places,round-precision=0] % integers
  @{\hspace{0.8em}}                                       % extra gap
  S[table-format=4.2,round-mode=places,round-precision=2] % 2 decimals
  S[table-format=4.2,round-mode=places,round-precision=2]
}
\toprule
Method & {\makecell{Succ.\\attempts}} & {$\lambda\!\sum\!\theta$} & {LQG cost} & {LQG+comm} \\
\midrule
ONE-SHOT & 12 & 2300 & 4892.956 & 7192.956 \\
\bl{MPC}      &  \bl{8} & \bl{1500} & \bl{4135.19} & \bl{5635.19} \\
\bottomrule
\end{tabular}
\end{table} \vspace{-.1cm}
\begin{table}[!h]
\centering
\setlength{\tabcolsep}{4pt}          
\renewcommand{\arraystretch}{1.05}  
\footnotesize
\caption{Averaged performance over 100 random seeds}
\label{tab:comm-costs}
\begin{tabular}{
  l
  S[table-format=2.0,round-mode=places,round-precision=2] 
  S[table-format=4.0,round-mode=places,round-precision=0] 
  @{\hspace{0.8em}}                                       
  S[table-format=4.2,round-mode=places,round-precision=2] 
  S[table-format=4.2,round-mode=places,round-precision=2]
}
\toprule
{Method} & {\makecell{Succ.\\attempts}} & {$\lambda\!\sum\!\theta$} & {LQG cost} & {LQG+comm} \\
\midrule
ONE-SHOT & 16.07 & 2328 & 5044.98 & ~~7372.98 \\
\bl{MPC}      & ~~~\bl{10.19} & \bl{1469} & \bl{4340.34} & \bl{5809.30} \\
\bottomrule
\end{tabular}
\end{table}

\begin{figure}
  \centering
  \includegraphics[width=0.75\columnwidth,trim=160 320 170 335]{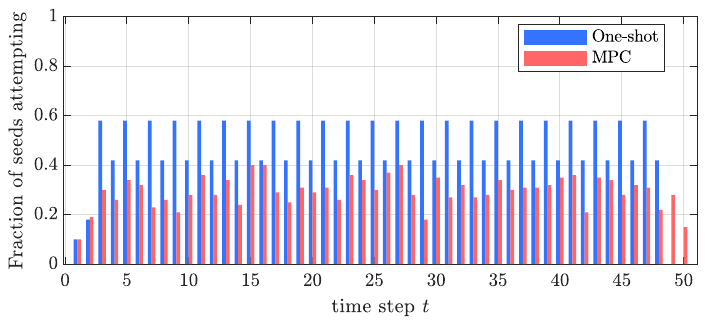}
  \caption{Normalized fraction of seeds attempting communication over time for one–shot and MPC strategies
  (\(p = 0.7\), \(\lambda = 100\)).}
  \label{fig:attempts_over_time}
\end{figure}

\begin{figure}[t]
\centering
\subfloat[Average number of attempts \label{fig:attempts_vs_p}]{
\includegraphics[width=0.48\columnwidth,trim=195 300 190 310,clip]{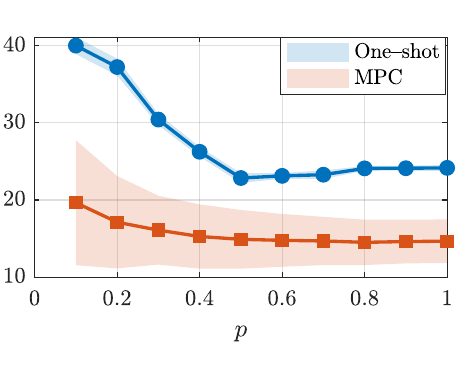}}
\hfill
\subfloat[Total realized cost \label{fig:total_cost_vs_p}]{
  \includegraphics[width=0.48\columnwidth,trim=195 300 190 310,clip]{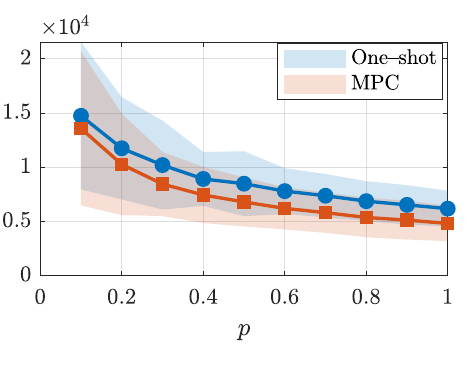}}
\caption{Comparison of one–shot and MPC schedulers versus channel success probability \(p\). Solid lines denote the mean; shaded bands show \(\pm 1\) standard deviation over 100 independent Monte Carlo runs.}
\label{fig:summary_vs_p}
\end{figure}

Fig.~\ref{fig:summary_vs_p} compares the one-shot and MPC-based schedulers as the channel success probability \(p\) varies.  
In Fig.~\ref{fig:attempts_vs_p}, the MPC scheduler consistently uses fewer transmission attempts than the one-shot schedule. In this figure, the one-shot scheduler’s attempts decrease with increasing packet success probability $p$ up to about $0.5$, then rises slightly. 
When $p$ is low, transmissions are frequent to compensate high loss.
As reliability improves, fewer transmissions are needed, but beyond $p \approx 0.5$ the scheduler becomes confident that attempts will succeed, leading to a slight increase in transmission activity.
This behavior indicates that MPC anticipates favorable transmission windows and avoids unnecessary communication when channel reliability is low. 
Correspondingly, Fig.~\ref{fig:total_cost_vs_p} shows that the total realized cost, which combines control performance and communication usage, 
also decreases as \(p\) increases. 
Across all probabilities, MPC achieves a lower overall cost, 
demonstrating the efficiency of predictive scheduling in balancing control performance and bandwidth utilization in networked control systems.

\section{CONCLUSIONS}\label{sec:conclusion}
We introduced a linear reformulation of event-triggered LQG over erasure channels by deriving a closed-form expression for the error covariance and recasting scheduling as a compact MILP. The approach accommodates heterogeneous communication penalties through a single effective charge and provides one-step send/skip certificates that obviate solving the MILP online; it also adds schedule-dependent upper and lower ratio bounds that define a performance envelope relative to a lossless channel. In a numerical case study, the MILP-based MPC policy achieved consistent reductions in both total cost and communication rate compared with a one-shot schedule across all packet-success probabilities examined. Future work will address output-only sensing with estimator co-design, multi-loop sharing, and infinite-horizon formulations.

\bibliographystyle{ieeetr}
\bibliography{References}

@inproceedings{molin2009lqg,
  title     = {On {LQG} Joint Optimal Scheduling and Control Under Communication Constraints},
  author    = {Molin, Adam and Hirche, Sandra},
  booktitle = {Proceedings of the 48th IEEE Conference on Decision and Control (CDC)},
  pages     = {5832--5838},
  year      = {2009},
  organization = {IEEE}
}

@article{molin2012optimality,
  title   = {On the Optimality of Certainty Equivalence for Event-Triggered Control Systems},
  author  = {Molin, Adam and Hirche, Sandra},
  journal = {IEEE Transactions on Automatic Control},
  volume  = {58},
  number  = {2},
  pages   = {470--474},
  year    = {2012},
  publisher = {IEEE}
}

@book{boyd2004convex,
  title     = {Convex Optimization},
  author    = {Boyd, Stephen P. and Vandenberghe, Lieven},
  year      = {2004},
  publisher = {Cambridge University Press}
}

@article{hashemi2025linear,
  title={A linear programming framework for optimal event-triggered LQG control},
  author={Hashemi, Zahra and Maity, Dipankar},
  journal={IEEE Control Systems Letters},
  volume={9},
  pages={2783--2788},
  year={2025},
  publisher={IEEE}
}

@article{soleymani2021value,
  title   = {Value of Information in Feedback Control: Quantification},
  author  = {Soleymani, Touraj and Baras, John S. and Hirche, Sandra},
  journal = {IEEE Transactions on Automatic Control},
  volume  = {67},
  number  = {7},
  pages   = {3730--3737},
  year    = {2021},
  publisher = {IEEE}
}

@phdthesis{thelander2020lqg,
  title  = {On {LQG}-Optimal Event-Based Sampling},
  author = {Thelander Andr{\'e}n, Marcus},
  school = {Lund University},
  year   = {2020}
}

@inproceedings{heemels2012introduction,
  title     = {An Introduction to Event-Triggered and Self-Triggered Control},
  author    = {Heemels, Wilhelmus P. M. H. and Johansson, Karl Henrik and Tabuada, Paulo},
  booktitle = {2012 IEEE 51st Conference on Decision and Control (CDC)},
  pages     = {3270--3285},
  year      = {2012},
  organization = {IEEE}
}

@article{kostina2019rate,
  title   = {Rate-Cost Tradeoffs in Control},
  author  = {Kostina, Victoria and Hassibi, Babak},
  journal = {IEEE Transactions on Automatic Control},
  volume  = {64},
  pages   = {4525--4540},
  year    = {2019},
  publisher = {IEEE}
}

@article{rabi2012adaptive,
  title   = {Adaptive Sampling for Linear State Estimation},
  author  = {Rabi, M. and Moustakides, G.~V. and Baras, J.~S.},
  journal = {SIAM Journal on Control and Optimization},
  volume  = {50},
  number  = {2},
  pages   = {672--702},
  year    = {2012}
}

@article{lipsa2011remote,
  title   = {Remote State Estimation with Communication Costs for First-Order {LTI} Systems},
  author  = {Lipsa, G.~M. and Martins, N.~C.},
  journal = {IEEE Transactions on Automatic Control},
  volume  = {56},
  number  = {9},
  pages   = {2013--2025},
  year    = {2011}
}

@article{mamduhi2025network,
  title   = {Network-Aware Optimal Sampling for Stochastic Control Systems over Dynamic Networks},
  author  = {Mamduhi, Mohammad H. and Maity, Dipankar},
  journal = {IEEE Control Systems Letters},
  volume  = {9},
  pages   = {1808--1813},
  year    = {2025},
  publisher = {IEEE}
}

@article{maity2019optimal,
  title   = {Optimal Event-Triggered Control of Nondeterministic Linear Systems},
  author  = {Maity, Dipankar and Baras, John S.},
  journal = {IEEE Transactions on Automatic Control},
  volume  = {65},
  number  = {2},
  pages   = {604--619},
  year    = {2019},
  publisher = {IEEE}
}

@article{maity2021multiagent,
  title   = {Multiagent Consensus Subject to Communication and Privacy Constraints},
  author  = {Maity, Dipankar and Tsiotras, Panagiotis},
  journal = {IEEE Transactions on Control of Network Systems},
  volume  = {9},
  number  = {2},
  pages   = {943--955},
  year    = {2021},
  publisher = {IEEE}
}

@article{maity2021optimalb,
  title   = {Optimal {LQG} Control of Networked Systems Under Traffic-Correlated Delay and Dropout},
  author  = {Maity, Dipankar and Mamduhi, Mohammad H. and Hirche, Sandra and Johansson, Karl H.},
  journal = {IEEE Control Systems Letters},
  volume  = {6},
  pages   = {1280--1285},
  year    = {2021},
  publisher = {IEEE}
}

@article{suthar2025fly,
  title={Where to Fly, What to Send: Communication-Aware Aerial Support for Ground Robots},
  author={Suthar, Harshil and Maity, Dipankar},
  journal={arXiv preprint arXiv:2512.06207},
  year={2025}
}
\appendices
\section{Proof of Theorem~\ref{thm:one-step-certificates-erasure}}
\label{app:proof-one-step}
\begin{proof}
For a decision at time $k$, define the cost-to-go under the two actions
\begin{align*}
J_k^{\text{skip}}(\Theta)
&:= \mathbb{E}\!\left[\sum_{t=k}^{T-1}\!\bigl(\|e_t\|_{\Gamma_t}^2+\lambda \theta_t\bigr)\,\middle|\,\theta_k=0,\ e_k^s\right],\\
J_k^{\text{att}}(\Theta)
&:= \mathbb{E}\!\left[\sum_{t=k}^{T-1}\!\bigl(\|e_t\|_{\Gamma_t}^2+\lambda \theta_t\bigr)\,\middle|\,\theta_k=1,\ e_k^s\right],
\end{align*}
and their optima
\[
J_k^{\text{skip},\star}:=\min_{\Theta}J_k^{\text{skip}}(\Theta),\qquad
J_k^{\text{att},\star}:=\min_{\Theta}J_k^{\text{att}}(\Theta).
\]
Set the (skip minus attempt) benefit
\[
\mathrm{Ben}_k(e_k^s):=J_k^{\text{skip},\star}-J_k^{\text{att},\star}.
\]
Bounding by evaluating at the opposite policies gives
\begin{equation}
\label{eq:Ben-brackets}
J_k^{\text{skip},\star}-J_k^{\text{att}}(\Theta^{\text{skip}})
\;\le\;
\mathrm{Ben}_k(e_k^s)
\;\le\;
J_k^{\text{skip}}(\Theta^{\text{att}})-J_k^{\text{att},\star},
\end{equation}
where {\small $\Theta^{\text{skip}}\in\arg\min_\Theta J_k^{\text{skip}}(\Theta)$} and
{\small $\Theta^{\text{att}}\in\arg\min_\Theta J_k^{\text{att}}(\Theta)$}.

At stage $k$, skipping incurs a cost/penalty $\|e_k^s\|_{\Gamma_k}^2$, whereas attempting yields
$(1-p)\|e_k^s\|_{\Gamma_k}^2+\lambda$ (success with probability $p$ resets the error to zero).
Thus the immediate expected gain from attempting is
\[
p\,e_k^{s\top}\Gamma_k e_k^s-\lambda.
\]

Unrolling the error for $j\ge 1$,
\[
e_{k+j}=\mu_{k+j,k+1}\,A^{j}e_k^s+\eta_{k+j},
\]
where $\tilde\theta_s:=\theta_s\delta_s$ denotes a successful reception, 
$\mu_{k+j,k+1}:=\prod_{s=k+1}^{k+j}(1-\tilde\theta_s)\in\{0,1\}$ indicates no success on $(k,k+j]$, 
and $\eta_{k+j}$ collects process-noise terms (independent of $e_k^s$ and zero-mean):
\[
\eta_{k+j} := \sum_{r=0}^{j-1}\mu_{k+j,k+1+r}\,A^{\,j-1-r} w_{k+r}.
\]
Taking expectations, the \emph{difference} of the tail costs depends only on the deterministic $e_k^s$ term and the indicators $\mu$. Evaluating the right-hand side of~\eqref{eq:Ben-brackets} at $\Theta^{\text{att}}$ yields
\begin{align*}
    J_k^{\text{skip}}(\Theta^{\text{att}})-J_k^{\text{att},\star}
= &\; p\,e_k^{s\top}\Gamma_k e_k^s-\lambda\\
  &\; + p\sum_{j=1}^{T-1-k}\mu_{k+j,k+1}^{\text{att}}\,
     \|e_k^s\|_{A^{j\!\top}\Gamma_{k+j}A^j}^2 .
\end{align*}
The worst (largest) tail arises if no success occurs after $k$, i.e.,
$\mu_{k+j,k+1}^{\text{att}}=1$ for all $j$, giving
\[
\mathrm{Ben}_k(e_k^s)\;\le\;
p\,e_k^{s\top}\,
\underbracket{\Bigl(\Gamma_k+\sum_{j=1}^{T-1-k}\!A^{j\!\top}\Gamma_{k+j}A^j\Bigr)}_{W_k}
\,e_k^s-\lambda.
\]

Similarly, evaluating the left-hand side of~\eqref{eq:Ben-brackets} at $\Theta^{\text{skip}}$ gives
\begin{align*}
    J_k^{\text{skip},\star}-J_k^{\text{att}}(\Theta^{\text{skip}})
= &\; p\,e_k^{s\top}\Gamma_k e_k^s-\lambda\\
  &\; + p\sum_{j=1}^{T-1-k}\mu_{k+j,k+1}^{\text{skip}}\,
     \|e_k^s\|_{A^{j\!\top}\Gamma_{k+j}A^j}^2 .
\end{align*}
The tightest (largest) lower bound occurs if an immediate success follows $k$, i.e.,
$\mu_{k+j,k+1}^{\text{skip}}=0$ for all $j$, which yields
\[
\mathrm{Ben}_k(e_k^s)\ge p\,e_k^{s\top}\Gamma_k e_k^s-\lambda.
\]
Combining the two bounds establishes~\eqref{eq:sandwich-erasure}.
\end{proof}

\section{Proof of Theorem~\ref{thm:ratio-lossless}}
\label{app:proof-ratio-lossless}
\begin{proof}
Fix a schedule $\Theta$. Recall
\begin{align}
\label{eq:Jp-def}
J_p(\Theta)
&= \lambda \sum_{t=k}^{T-1} \theta_{t|k}
+ \sum_{t=k}^{T-1}\sum_{\tau=k}^{t} \mathbf{1}\!\{c_{t,\tau}(\Theta)=0\}\, g_{t,\tau} \nonumber\\
&\quad+ \sum_{t=k}^{T-1}\sum_{\tau=k}^{t} \mathbf{1}\!\{c_{t,\tau}(\Theta)\neq 0\}\,(1-p)^{c_{t,\tau}(\Theta)}\, g_{t,\tau},
\end{align}
and the lossless ($p=1$) cost (cf.~\cite{hashemi2025linear})
\begin{align}
\label{eq:J1-def}
J_1(\Theta)
&= \sum_{t=k}^{T-1}\sum_{\tau=k}^{t} \mu_{t,\tau}\, g_{t,\tau} \;+\; \lambda \sum_{t=k}^{T-1} \theta_{t|k} \nonumber\\
&= \sum_{t=k}^{T-1}\sum_{\tau=k}^{t} \mathbf{1}\!\{c_{t,\tau}(\Theta)=0\}\, g_{t,\tau}
\;+\; \lambda \sum_{t=k}^{T-1} \theta_{t|k} .
\end{align}
Subtracting \eqref{eq:J1-def} from \eqref{eq:Jp-def} gives
\begin{equation}
\label{eq:Jdiff-app}
J_p(\Theta) - J_1(\Theta)
= \sum_{t=k}^{T-1}\sum_{\tau=k}^{t}
\mathbf{1}\!\big\{c_{t,\tau}(\Theta)\neq 0\big\}\,
(1-p)^{\,c_{t,\tau}(\Theta)}\, g_{t,\tau}\;,
\end{equation}
because each $g_{t,\tau}\ge 0$ and $(1-p)^{c}\ge 0$.
% Taking minima over $\Theta$ on both sides of \eqref{eq:Jdiff-app} yields
% $J_p^{\star}\ge J_1^{\star}$.

\emph{Upper bound:}
Optimality implies $J_p^{\star}\le J_p(\Theta_1^{\star})$.
Using \eqref{eq:Jdiff-app} at $\Theta=\Theta_1^{\star}$ and the inequality
$(1-p)^{c}\le (1-p)$ for all $c\ge 1$,
\[
J_p(\Theta_1^{\star})-J_1^{\star}
\le (1-p)\sum_{t=k}^{T-1}\sum_{\tau=k}^{t}
\mathbf{1}\!\big\{c_{t,\tau}(\Theta^{\star}_{1})\ge 1\big\}\, g_{t,\tau}.
\]
Divide by $J_1^{\star}$ to obtain \eqref{eq:ratio-tight}.

\emph{Lower bound:}
Evaluate \eqref{eq:Jdiff-app} at $\Theta_p^{\star}$:
\[
J_p^{\star}-J_1(\Theta_p^{\star})
= \sum_{t=k}^{T-1}\sum_{\tau=k}^{t}
\mathbf{1}\!\big\{c_{t,\tau}(\Theta^{\star}_{p})\neq 0\big\}\,
(1-p)^{\,c_{t,\tau}(\Theta^{\star}_{p})}\, g_{t,\tau}.
\]
Because $0\le c_{t,\tau}(\Theta^{\star}_{p})\le t-\tau+1\le H$,
we have $(1-p)^{\,c_{t,\tau}(\Theta^{\star}_{p})}\ge $ $ (1-p)^{C_{\max}}$ with
$C_{\max}:=\max_{t,\tau}c_{t,\tau}(\Theta^{\star}_{p})$.
Thus,
\[
J_p^{\star}-J_1(\Theta_p^{\star})
\;\ge\;
(1-p)^{C_{\max}}
\sum_{t=k}^{T-1}\sum_{\tau=k}^{t}
\mathbf{1}\!\big\{c_{t,\tau}(\Theta^{\star}_{p})\ge 1\big\}\, g_{t,\tau}.
\]
Divide by $J_1(\Theta_p^{\star})$ and use $J_1(\Theta_p^{\star})\ge J_1^{\star}$
to get
\[
\frac{J_p^{\star}}{J_1^{\star}}
\;\ge\;
\frac{J_p^{\star}}{J_1(\Theta_p^{\star})}
\;\ge\;
1 +
(1-p)^{C_{\max}}
\frac{\displaystyle \sum_{t=k}^{T-1}\sum_{\tau=k}^{t}
\mathbf{1}\!\big\{c_{t,\tau}(\Theta^{\star}_{p})\ge 1\big\}\, g_{t,\tau}}
{\displaystyle J_1(\Theta_p^{\star})},
\]
which is exactly \eqref{eq:ratio-lower-dependent}.
\end{proof}

\end{document}